\documentclass[twocolumn, showpacs,superscriptaddress]{revtex4-2}

\usepackage{multirow}
\usepackage{mathrsfs}
\usepackage{amsmath}
\usepackage[colorlinks=true, allcolors=blue]{hyperref}
\usepackage{verbatim}
\usepackage{amsfonts}
\usepackage{amssymb}
\usepackage{amsthm}
\usepackage{bm}
\usepackage{dsfont}
\usepackage{xparse}
\usepackage{graphicx}
\usepackage{xcolor}
\usepackage{textcase}
\usepackage{url}
\usepackage{lipsum}
\usepackage{epstopdf}
\usepackage{footnote}
\usepackage{color}

\newtheorem*{lemma*}{Lemma}
\newtheorem*{def*}{Definition}
\def\a{\alpha}
\def\b{\beta}
\def\c{\gamma}
\def\G{\Gamma}
\def\d{\delta}
\def\e{\epsilon}
\def\l{\lambda}

\def\m{\mu}
\def\s{\sigma}
\def\t{\tau}
\def\j{\varphi}
\def\f{\phi}
\def\q{\theta}
\def\he{H}
\def\ghe{\mathrm{H}}
\def\mqb{\overset{\mathrm{M}}{\longmapsto}}
\def\tr{\text{tr}}

\newcommand{\avg}[1]{\langle{#1}\rangle}
\newcommand{\ket}[1]{| {#1} \rangle} 
\newcommand{\bra}[1]{\langle {#1} |} 
\newcommand{\braket}[2]{\langle {#1} \vphantom{#2} | {#2} \vphantom{#1} \rangle}
\newcommand{\mel}[3]{\langle {#1} \vphantom{#2} | {#2} \vphantom{#3} | {#3} \rangle}
 
\DeclareDocumentCommand{\Tr}{m m O{\big}}{{\rm Tr}_{\:\!{#1}}#3({#2}#3)}
\DeclareMathOperator*{\ox}{\otimes}
\DeclareMathOperator*{\Var}{\mathrm{Var}}
\DeclareMathOperator*{\Span}{Span}
\DeclareMathOperator{\sgn}{sgn}
\newcommand{\bket}[1]{\big| {#1} \big\rangle}
\newcommand{\bbra}[1]{\big\langle {#1} \big|}

\usepackage{enumitem}
\usepackage{float}

\begin{document}
\title{Macroscopically nonlocal quantum correlations}
\author{Miguel Gallego}
\email{miguel.gallego.ballester@univie.ac.at}
\affiliation{University of Vienna, Faculty of Physics, Vienna Center for Quantum Science and Technology, Boltzmanngasse 5, 1090 Vienna, Austria}
\author{Borivoje Daki\'c}
\email{borivoje.dakic@univie.ac.at}
\affiliation{University of Vienna, Faculty of Physics, Vienna Center for Quantum Science and Technology, Boltzmanngasse 5, 1090 Vienna, Austria}
\affiliation{
Institute for Quantum Optics and Quantum Information (IQOQI),
Austrian Academy of Sciences, Boltzmanngasse 3,
A-1090 Vienna, Austria}

\date{\today}

\begin{abstract}
It is usually believed that coarse-graining of quantum correlations leads to classical correlations in the macroscopic limit. Such a principle, known as macroscopic locality, has been proved for correlations arising from independent and identically distributed (IID) entangled pairs. In this letter we consider the generic (non-IID) scenario. We find that the Hilbert space structure of quantum theory can be preserved in the macroscopic limit. This leads directly to a Bell violation for coarse-grained collective measurements, thus breaking the principle of macroscopic locality.

\end{abstract}

\maketitle
\subsection*{Introduction}
Quantum mechanics does not impose any limit on the size of the system it describes, which can, in principle, be as large as a cat \cite{schrodinger}. However, quantum behavior is not observed at the macroscopic scale, where the world appears to be classical. The idea that quantum mechanics must reproduce classical physics in the limit of large quantum numbers is known as the correspondence principle \cite{bohr}. Yet, this principle, in all its generality, has not been rigorously stated and proved, mostly because the concept of ``macroscopic" remains somewhat vague. While different interpretations of the macroscopic limit may lead in general to different conclusions, there is still confidence that quantum behavior must somehow disappear in the limit. 

A possible explanation for the emergence of classicality from quantum theory is via the coarse-graining of the measurements \cite{mermin, poulin, caslav, nw, kofler2013, jeong2014, henson2015, demelo}. In this respect, one important consequence of the correspondence principle is the concept of \textit{macroscopic locality} (ML) \cite{nw}: coarse-grained quantum correlations become local (in the sense of Bell \cite{bell}) in the macroscopic limit. ML has been challenged in different circumstances, both theoretically and experimentally \cite{bancalmultipair, tura, schmied, poh, manybox, paritybellviolation, cabello, meng} (see Ref. \cite{daltonreview} for a review). However, as far as we know, nonlocality fades away under coarse-graining when the number of particles $N$ in the system goes to infinity. In this sense, ML was proposed by Navascu\'{e}s and Wunderlich (NW) \cite{nw} as an axiom for discerning physical post-quantum theories. In particular, they considered a bipartite Bell-type experiment where the parties measure intensities with a resolution of the order of $\sqrt{N}$, or equivalently $O(\sqrt{N})$ coarse-graining. Then, under the IID assumption, i.e. under the premise that particles are only entangled by \emph{independent and identically distributed} pairs, they prove ML for quantum theory.

In this letter we generalize the concept of ML to any level of coarse-graining $\a \in [0,1]$, meaning that the intensities are measured with a resolution of the order of  $N^\a$. We drop the IID assumption and we investigate the existence of a boundary between quantum (nonlocal) and classical (local) physics, identified by the minimum level of coarse-graining $\a$ required to restore locality. To do this, we introduce the concept of \emph{macroscopic quantum behavior} (MQB), demanding that the Hilbert space structure, such as the \emph{superposition principle}, is preserved in the thermodynamic limit. Then, we provide a concrete example of MQB at $\a = 1/2$ which violates ML. This is the opposite of what happens in the IID case, where ML is known to hold, as shown by NW. Finally, we analyze the effects of noise and particle losses, showing robustness of the macroscopic statistics. Altogether, our findings shed new light on the problem of the transition (if any) between quantum and classical physics.

\subsection*{Experimental setup and macroscopic locality}
We consider a simple Bell-type setting as illustrated in Fig. \ref{fig:mbe}. A state $\rho^{[2N]}$ of $2N$ particles is produced, out of which $N$ are sent to Alice and $N$ to Bob. Alice performs a collective measurement described by the (single-particle) POVM elements $E_{a|p}^{A}$, where $a \in \Omega_A$ is her outcome and $p \in \Sigma_A$ is her measurement setting (and similarly does Bob). Alice (Bob) has the following limitations:
\begin{enumerate}[label=(\textit{\roman*})]
\item \label{ass:intensity} \textit{Intensity measurement}: No access to individual outcomes, but only to their sum or \textit{intensity} $I_A = \sum_{i=1}^N a_i$. 
\item \label{ass:resolution} $O(N^\a)$ \textit{coarse-graining}: The measuring scale for $I_{A}$ has a limited resolution of the order of $N^\a$, where $\a \in [0, 1]$ is the \emph{order} or \emph{level of coarse-graining}. 
\end{enumerate}
These assumptions naturally lead to the following macroscopic variable
\begin{align}
    X_\a^{[N]} = \frac{1}{N^\a}   \sum_{i=1}^N  \big( a_i - \avg{a_i} \big) \, .
    \label{eq:macrvar}
\end{align}
This quantity, with mean set to zero, has well-studied limit properties in the classical domain \cite{gk}. The special case $\a = 1/2$ (also discussed here) is closely related to the the central limit  theorem \cite{clt}, for which $X_{\a=1/2}^{[N]}$ properly captures  the quantum  fluctuations  of  the  intensity  about  its  mean \cite{benatti2014}. The above macroscopic variable in turn defines the POVM associated to it
\begin{align}\label{CG:povm}
    E(X_\a^{[N]}) = \sum_{\sum_{i=1}^N  ( a_i - \avg{a_i} ) / N^\a \, = \, X_\a^{[N]}} E_{a_1}^A \ox ... \ox E_{a_N}^A \, .
\end{align}
Altogether, $X_\a^{[N]}$ and $E\big(X_\a^{[N]}\big)$ specify Alice's measurement, and likewise holds for Bob. Finally, Alice and Bob repeat their experiment many times in order to extract the bipartite distribution. The central quantity of interest is the limit thereof, namely
\begin{align}\label{eq:biplimdistribution}
    P(x, y) = \lim_{N \to \infty} \tr \,  \rho^{[2N]}E\big(X_\a^{[N]}\big) \otimes E\big(Y_\a^{[N]}\big) \, ,
\end{align}
where $X_\a^{[N]}\to x$ and $Y_\a^{[N]}\to y$ denote convergence in distribution. A necessary condition for convergence is that the variance scaling of the measured intensity (as determined by the state $\rho^{[2N]}$) matches the order of coarse-graining, i.e. that if the variance of the intensity $\Var(I)$ scales as $N^{2 \b}$, then $\a=\b$. Otherwise, if $\a < \b$, the distribution \eqref{eq:biplimdistribution} will simply not converge, and if  $\a > \b$, the distribution will converge to a Dirac delta function, thus giving trivial statistics.

The question of ML refers to the locality properties of the limit distribution \eqref{eq:biplimdistribution}. We will say a theory possesses ML at order $\a$ if the limit distributions $P(x, y)$ for any choice of measurements can be described by a local model
\begin{align*}
    P(x, y) = \int d \lambda \, \mu(\lambda) \, P_A(x | \lambda) \, P_B(y | \lambda) \, .
\end{align*}
On the other hand, if the above factorization does not hold, we say that the theory exhibits \emph{macroscopically nonlocal correlations} (at order $\a$). It seems natural to conjecture the existence of a quantum-to-classical transition point, i.e. a critical value $\a_c$ such that quantum theory violates ML at any $\a<\a_c$, while locality is restored for $\a>\a_c$. Note that, at $\a=\a_c$, both ML and violation of ML are possible. Intuitively, we expect quantum theory to violate ML at $\a=0$ (no coarse-graining) and to satisfy it at $\a=1$ (full coarse-graining). Indeed, there are strong evidences that this is the case, as presented in Refs. \cite{paritybellviolation} and \cite{barbosa} respectively. A more interesting result is that of NW. In their paper, the authors consider the case of $O(\sqrt{N})$ coarse-graining ($\a=1/2$) with IID states, described by a density matrix of the form $\rho^{[2N]} = (\rho_{AB})^{\ox N}$. Then, according to the central limit theorem, the distributions of the macroscopic variables \eqref{eq:macrvar} are Gaussian. Using this, they show that the corresponding bipartite Gaussian distributions \eqref{eq:biplimdistribution} are local, implying that $\a_c^\text{IID} \leq 1/2$. However, it is far from clear whether actually $\a_c^\text{IID} = 1/2$, and the possibility that even $\a_c^\text{IID} =0$ is not discarded \cite{manybox}. The IID assumption might then be too restrictive. Our goal here is to drop it and consider the most general scenario. In the next section, we show how non-IID states at $\a=1/2$ can give rise to non-Gaussian distributions in the limit. This result will lead us to full quantum behavior at $\a=1/2$, for which we also show violation of ML. Thus we prove that $\a_c \geq 1/2$ in the general case (non-IID).

\begin{figure}
        \centering
        \includegraphics[width=0.45\textwidth]{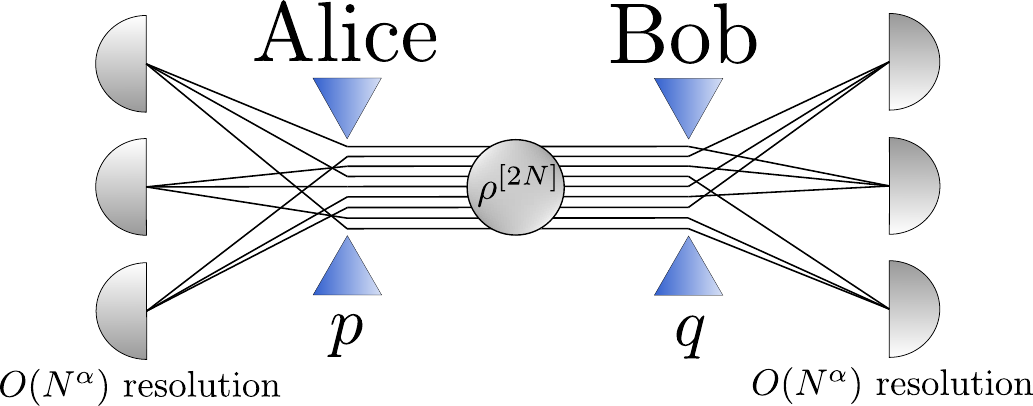}
        \caption{\textbf{Macroscopic Bell-type experiment.} A source produces a $2N$-particle state, and sends half of the particles to Alice and half to Bob. The parties perform collective measurements on their beams, specified by the settings $(p,q)$. For each outcome, local detectors count the number of particles with resolution of order $N^\a$.}
        \label{fig:mbe}
\end{figure}

\subsection*{Non-IID states and non-Gaussian limit distributions}
In the following, we will consider the case $\a=1/2$ for a single party, say Alice, and a system of $N$ spin-$1/2$ particles (qubits). For simplicity, we will consider projective (von Neumann) measurements, with $E_a^2 = E_a$, for which we choose binary outcomes $a=\pm 1$. It is useful to define the (single-particle) observable $A= \sum_a a E_a$, such that the macroscopic variable \eqref{eq:macrvar} is naturally promoted to the macroscopic observable
\begin{align}
    \hat{X}_{\a=1/2}^{[N]} = \frac{1}{\sqrt{N}} \sum_{i=1}^N \Big( A_i - \avg{A_i} \Big) \, .
    \label{eq:macrobs}
\end{align}
We consider an example of non-IID state, the W state
\begin{align*}
    \ket{W} = \frac{1}{\sqrt{N}} \Big( \ket{100...0} + \ket{010...0}+\dots +\ket{000...1} \Big) \, .
\end{align*}
The characteristic function of $\hat{X}_{\a=1/2}^{[N]}$ for this state can be written in terms of $\mathcal{A} = e^{i t (A- \avg{A})/\sqrt{N}}$ as $\chi(t) =   \mel{W}{ \mathcal{A}^{\ox N}}{W}$. Direct computation gives
\begin{align*}
\chi(t) =  \mathcal{A}_{00}^{N-2} \,  \Big[ \mathcal{A}_{11} \mathcal{A}_{00} +  (N-1) \mathcal{A}_{10}  \mathcal{A}_{01} \Big] \, ,
\end{align*}
where $\mathcal{A}_{ij}= \mel{i}{\mathcal{A}}{j}$. Now, expanding $\mathcal{A} = 1 + it(A-\avg{A}) /\sqrt{N} - t^2 (A-\avg{A})^2/N + O(N^{-3/2})$ and using that $\avg{A}=(1-1/N)A_{00}+1/N A_{11}$, the thermodynamic limit reads $\chi(t) = e^{-\s^2 t^2/2} \big( 1 -  \s^2 t^2 \big)$, where we have defined the variance $\s^2 = \avg{A^2}-\avg{A}^2$. The Fourier transform of $\chi(t)$ gives the non-Gaussian limit distribution
\begin{align*}
    P(x) = \frac{1}{\sqrt{2 \pi \s^2}} \,  \frac{x^2}{\s^2} \, e^{-x^2/(2 \s^2)} \, .
\end{align*}
Similarly, for any $N$-particle Dicke state \cite{dicke}
\begin{align*}
\ket{N,k} = \frac{1}{\sqrt{{N \choose k}}} \Big( \ket{ \underbrace{1...1}_{k} 0...0 } + \text{permutations} \Big) \, ,
\end{align*}
it has been shown in Ref. \cite{dorlas} that the limit distribution of macroscopic observables of the form of \eqref{eq:macrobs} can be written in terms of Hermite polynomials as $P_k(x) \sim e^{-x^2/2} \he_k^2(x)$. Thus, for such a family of states, we find non-Gaussian limiting behavior, which is conventionally associated to non-classical phenomena. Moreover, such a distribution coincides with the Born rule for the distribution in position of the $k$-th excited state of the harmonic oscillator, $P(x) = | \langle x | k \rangle |^2$, where the wave-function is
\begin{align*}
    \langle x | k \rangle = \frac{1}{(2 \pi )^{1/4}} \frac{1}{\sqrt{k!}} e^{-x^2/4} \, \he_k (x ) \, .
\end{align*}
Here, $H_k(x)$ are the Hermite polynomials \footnote{Note that we are not using the standard convention for Hermite polynomials.}
\begin{align*}
    \he_k(x) = (-1)^k e^{x^2/2} \frac{d^k}{dx^k} e^{-x^2/2} \, .
\end{align*}

\subsection*{Macroscopic quantum behavior}
The analysis provided above suggests the following limit mapping
\begin{align}
\ket{N,k} \underset{N \to \infty}{\xrightarrow{\hspace*{.8cm}}}  \ket{k} \, , \quad \hat{X}_{\a=1/2}^{[N]} \underset{N \to \infty}{\xrightarrow{\hspace*{.8cm}}} \hat{x} \, ,
\label{eq:nonlinearmap}
\end{align}
where $\ket k$ is the number basis of the harmonic oscillator and $\hat{x}$ is the position operator. The identification between Dicke states and eigenstates of the harmonic oscillator is natural: it is known that the $SU(2)$ algebra of angular momentum contracts to the Heisenberg algebra of creation and annihilation operators in the limit of large total angular momentum \cite{dickealgebra}. This is also related to the so-called photon-spin mapping \cite{photonspinmapping}. However, \eqref{eq:nonlinearmap} is not just a map of states, but a joint map of states and observables. We would like to formalize this idea and generalize it to general POVM measurements and to any level of coarse-graining. We shall define a mapping that preserves the (separable) Hilbert space structure of quantum theory, i.e. the \emph{superposition principle} as well as the \emph{Born rule}. To achieve this, let us introduce the concept of \textit{macroscopic quantum behavior} (MQB) in the following way.
\begin{def*}
Let $\mathcal{H}^N$ be the Hilbert space of $N$ particles. For every $N \geq N_0$, let $\mathcal{M}_d^N \subset \mathcal{H}^N$ be a subspace of fixed dimension $d \geq 2$, and let $\mathcal{M}_d^\infty$ be an auxiliary $d$-dimensional Hilbert space. Since these (sub)spaces are all $d$-dimensional, they are all isomorphic as vector spaces:
\begin{align*}
    \mathcal{M}_d^{N_0} \cong \mathcal{M}_d^{N_0+1} \cong \mathcal{M}_d^{N_0+2} \cong ... \cong \mathcal{M}_d^\infty \, .
\end{align*}
Let us fix the sequence of isomorphisms by a choice of basis in every space:
\begin{align*}
    \ket{k}_{N_0} \mapsto \ket{k}_{N_0+1} \mapsto \ket{k}_{N_0+2} \mapsto ...\mapsto  \ket{k}_\infty \, ,  
\end{align*}
for every $k=1, ..., d$. This gives a unique identification among states $\ket{\Psi_{N}} = \sum_{k=1}^d c_k \ket{k}_{N} \in \mathcal{M}_d^{N}$ for all $N \geq N_0$, including $\ket \psi = \sum_{k=1}^d c_k \ket{k}_\infty \in \mathcal{M}_d^\infty$.

Now, we say that the sequence of spaces $\mathcal{M}_d^N$, together with the corresponding choice of bases (isomorphisms), possesses MQB at order $\a$ if, for any $\ket{\Psi_{N}} \in \mathcal{M}_d^{N}$, we have
\begin{align}
    \lim_{N\to \infty} \mel{\Psi_N}{E\big( X_\a^{[N]} \big) }{\Psi_N} = \mel{\psi}{e(x)}{\psi}
    \label{eq:limitcondition}
\end{align}
for all measurements specified by \eqref{eq:macrvar} and \eqref{CG:povm}. Here, $X_\a^{[N]} \to x \in \Omega$ and $e(x)$ is a POVM element acting in $\mathcal{M}_d^\infty$, satisfying $\sum_{x \in \Omega} e(x) = 1$. 
\end{def*}
This definition clearly ensures that both the Born rule and the superposition principle remain valid in the macroscopic limit. In this case we write
\begin{align*}
    \ket{k}_N \mqb \ket k \, , \quad E \big(X_\a^{[N]} \big) \mqb e(x) \, .
\end{align*}
In order to fulfill the above definition, the macroscopic variable as defined in Eq. \eqref{eq:macrvar} needs to be modified, as it exhibits a nonlinear dependence on the input state via the mean value $\avg{a_i}$. Such a behavior is inconsistent with the Born rule in general, and the simplest way to fix this is to substitute this mean value by some constant $\m$, thus redefining
\begin{align}\label{eq:macrvarlinear}
    X_\a^{[N]} = \frac{1}{\t \, N^\a} \sum_{i=1}^N \big( a_i - \m \big) \, ,
\end{align}
where the parameter $\t$ is introduced for future convenience. Such a modification of $X_\a^{[N]}$ consequently induces a modification of the POVM element $E\big(X_\a^{[N]} \big)$ given in Eq. \eqref{CG:povm}. With this we are ready to provide examples.

\subsection*{MQB at $\a=1/2$}
We consider again Dicke states, i.e. the sequence of spaces $\mathcal{M}^N = \Span \{ \ket{N,k} \}_k$ together with the identification $\ket{N,k} \mapsto \ket{N+1,k}$, where $k=0, 1, ... , d-1$ for some finite $d$ which we leave unspecified for now. We shall directly evaluate the limit \eqref{eq:limitcondition}. Given a single-particle POVM with elements $E_a$, let us define the matrix $A = \sum_a a E_a$. The explicit calculation of the distribution of $X_{\a=1/2}^{[N]}$ on a state $\ket{\Psi_N} = \sum_k c_k \ket{N,k}$ is provided in Appendix \ref{appa}. Choosing $\mu = A_{00}$ and $\tau  = |A_{01}|$ for the macroscopic variable \eqref{eq:macrvarlinear}, the limit distribution can be written as
\begin{align}
    P(x) = \sum_{k,l} e^{-i k \j} c_k^* c_l  e^{i l \j}  \int dx' \, \frac{e^{-  \frac{(x-x')^2}{2 s^2} }}{\sqrt{2 \pi s^2}} \braket{k}{x'} \braket{x'}{l} \, ,
    \label{eq:distributionnk}
\end{align}
where $\j = \arg (-A_{01})$, $s^2 = \s^2/\t^2-1$ in terms of the limit variance $\s^2 = \mel{0}{\sum_a a^2 E_a}{0} - \mel{0}{A}{0}^2$, $\ket x$ is the position basis and $\ket k$ is the number basis of the harmonic oscillator. As we see, for the limit variable we get $x \in \mathbb{R}$. Therefore, the MQB at $\a=1/2$ is given by
\begin{align}
\ket{N,k} \mqb \ket{k} \, , \quad   E\big( X_{\a=1/2}^{[N]} \big) \mqb U^\dagger_\j \, e_s(x) U_\j \, ,
 \label{eq:linearmap}
\end{align}
where $U_\j= e^{i \j \hat{k}}$, in terms of the number operator $\hat{k}$ for the harmonic oscillator, and $e_s(x)$ is the Gaussian POVM element
\begin{align}
    e_s(x) = \frac{1}{\sqrt{2 \pi  s^2}} \int dx' \, e^{- \frac{(x-x')^2}{2 s^2}} \, \ket{x'} \bra{x'} \, .
    \label{eq:gaussianpovm}
\end{align}
We see that the limit (auxiliary) space $\mathcal{M}_d^\infty$ is naturally embedded in the infinite-dimensional space of the harmonic oscillator. Since the dimension $d$ is arbitrarily large, we can freely set $\mathcal{M}_d^{\infty}$ to be the whole space of the harmonic oscillator.

Let us now analyze in more detail the simple case of projective measurements. For these, the values of $\s$ and $\t$ coincide, so that $s=0$. Then $e_s(x)$ becomes the projector on position $\ket{x} \bra{x}$, and the observable $\hat{X}^{[N]}= \sum_{X^{[N]}} X^{[N]} \, E(X^{[N]})$ becomes in the limit the phase-space observable
\begin{align}
U_\j^\dagger \, \hat{x} \, U_\j = \hat{x} \cos \j  + \hat{p} \sin \j \, ,
\label{eq:phasespace}
\end{align}
depending on $A$ only through the off-diagonal phase $\j = \arg ( - A_{01})$. With this, one can see that the operators $\hat{X}_{\a=1/2}^{[N]}$ form a non-commutative bosonic algebra in the thermodynamic limit. This has been known in the context of fluctuation observables \cite{goderis1989, benatti2014}. From the perspective of MQB, we obtain incompatibility of measurements along with the superposition principle in the macroscopic limit. This is a strong hint for violation of ML.

\subsection*{Violation of ML at $\a=1/2$}
Now we consider again the bipartite Bell scenario depicted in Fig. \ref{fig:mbe}, and assume for simplicity projective measurements. Let the source produce a bipartite $2N$-particle state of the form
\begin{align}
\ket{\Psi_{2N}} = \sum_k c_k \ket{N,k}_A \ox \ket{N,k}_B \, .
\label{eq:diagonaldicke}
\end{align}
Let Alice measure the macroscopic observable $\hat{X}_{\a=1/2}^{[N]}$ with settings $p \in \Sigma_A$, and let Bob measure $\hat{Y}_{\a=1/2}^{[N]}$ with settings $q \in \Sigma_B$. Using the MQB \eqref{eq:linearmap} for both, these measurements become phase-space observables \eqref{eq:phasespace} for different angles $\varphi_A$ and $\varphi_B$ on a state
\begin{align*}
\ket{\psi} = \sum_k c_k \ket{k}_A \ox \ket{k}_B \, .
\end{align*}
Such a system exhibits Bell nonlocality for a suitable choice of the constants $c_k$ and phase-space measurements \cite{munro}, and can be easily generalized to the multipartite case. In Appendix \ref{appb} we show explicit violation of the CHSH inequality. 

\subsection*{Robustness}
In this part we will study robustness of our MQB. We will analyze this robustness in two ways: losses and noise at the microscopic level and global noise of the order of $\sqrt{N}$. First, we consider the case of losses, where individual particles only reach the detectors with some probability $p \in [0,1]$, and they are lost with probability $1-p$. If the parties are able to measure the number of received particles (with a precision of the order of $\sqrt{N}$), then, as shown in Appendix \ref{appc}, this loss simply translates into a rescaled variable $x \to x/p$ together with $\s^2 \to \s^2/p$. Rescaling back to the old variable, we get an effective broadening of the limit Gaussian POVM \eqref{eq:gaussianpovm}
\begin{align}
    s_p^2 = \frac{\s^2}{p^3 \t^2} - 1  \, .
    \label{broadwidth}
\end{align}
Similarly, independent single-particle noise channels $\rho^{[N]} \mapsto \Gamma^{\ox N} (\rho^{[N]})$ can be absorbed in the single-particle POVMs (see Appendix \ref{appc}). Since any POVM is mapped to the limit POVM $U_\j^\dagger e_s(x) U_\j$ parametrized by $s$ and $\j$, such noise channels can only affect the macroscopic statistics in two simple ways: coherently, by shifting the angle $\j$, or incoherently, by enlarging the width $s$. In Appendix \ref{appc} we provide explicit calculations for the depolarizing and dephasing channels, showing broadening effects as in \eqref{broadwidth}. Finally, we consider the measurement precision (both of the intensity and of the number of particles) to be captured by some classical independent noise bounded by $\tau \e \sqrt{N}$ for some constant $\e$. This simply translates into an additive classical random variable $r$ bounded by $\e$, so that the random variable in the MQB \eqref{eq:linearmap} becomes
\begin{align}
    X_{\a=1/2}^{[N]} \, \to \,  x + r\, .
    \label{eq:robustrv}
\end{align}
Altogether this shows that our MQB is robust. In principle, the macroscopic Bell violation should still be observable for small enough values of the global noise $\e$ and of the effective parameter $s$ capturing noise and losses at the microscopic level.

\subsection*{MQB at $\a=1$}
We close with a final example of MQB for $\a=1$ (maximal coarse-graining) with $2N$-particle Dicke states $\ket{2N, N+k}$. As before, we have $A= \sum_a a E_a$, and choose $\m = \frac{1}{2} \tr A$ and $ \t = |A_{01}|$ for the macroscopic variable \eqref{eq:macrvarlinear}. In Appendix \ref{appd} we compute the limit distribution $P(x)$ of the variable $X_{\a=1}^{[2N]}$ on a superposition state $\ket{\Psi_{2N}} = \sum_k c_k \ket{2N, N+k}$. This distribution has a finite support $x \in [-1, 1]$, as opposed to the case $\a=1/2$, where we had $x \in \mathbb{R}$. It is therefore convenient to set $x=\cos \q$, with $\q \in [0, \pi]$, so that the distribution can be written as
\begin{align}
    P(\q) = \sum_{k,l}  e^{-i k \j} \, c_k^* c_l \,  e^{i l \j}  \,  \frac{ e^{-i(k-l) \q} + e^{i(k-l) \q}}{ 2 \pi } \, ,
    \label{eq:distribution2nn}
\end{align}
where $\j = \arg(A_{01})$. Then, the sequence of spaces $\mathcal{M}^{2N} = \Span \{ \ket{2N, N+k} \}_k$ together with the identification $\ket{2N, N+k} \mapsto \ket{2(N+1), N+1+k}$ has MQB at $\a=1$ given by
\begin{align*}
  \ket{2N, N+k}  \mqb  \ket k \, , \quad   E \big(X_{\a=1}^{[2N]} \big) \mqb U^\dagger_\j \, e(x) U_\j \, .
\end{align*}
Here, $\ket k$ is the eigenbasis of the quantum rotor \footnote{The quantum rotor is the system with one rotational degree of freedom $\q$ and Hamiltonian $L_z = - i \partial_{\q}$.}, with wave-functions
\begin{align*}
    \braket{\pm \q}{k} = \frac{1}{\sqrt{2 \pi}} e^{\pm i k \q} \, , \quad \q \in [0, \pi] \, ,
\end{align*}
$U_\j = e^{i \j \hat{k}}$ and $e(\q)= \ket \q \bra \q + \ket{- \q} \bra{ - \q}$. Despite the MQB, the POVMs $U^\dagger_\j \, e(x) U_\j = \ket{\q-\j} \bra{\q-\j} + \ket{-\q-\j} \bra{-\q-\j}$ are compatible for all $\j$. Given this measurement compatibility we can construct a joint distribution for all settings, thus restoring classicality and Bell locality. In Appendix \ref{appd} we explicitly provide a local model for the bipartite distributions arising from this MQB. 

\subsection*{Conclusions}
In this letter we have introduced a generalized concept of \emph{macroscopic locality} at any level of coarse graining $\a \in [0,1]$. We have investigated the existence of a critical value $\a_c$ that marks the quantum-to-classical transition. We have introduced the concept of MQB at level $\a$ of coarse-graining, which implies that the Hilbert space structure of quantum mechanics is preserved in the thermodynamic limit. This facilitates the study of macroscopic quantum correlations. By means of a particular MQB at $\a=1/2$ we show that $\a_c \geq 1/2$, as opposed to the IID case, for which $\a_c^\text{IID} \leq 1/2$. An upper bound on $\a_c$ is however lacking in the general case. The possibility that no such transition exists remains open, and perhaps there exist systems for which ML is violated at $\a=1$.

\begin{acknowledgments}
\emph{Acknowledgments.}--- We would like to thank Nicol\'{a}s Medina-S\'{a}nchez 
and Joshua Morris for helpful comments. Both authors acknowledge support from  the Austrian Science Fund (FWF) through BeyondC-F7112. 
\end{acknowledgments}
\bibliographystyle{apsrev4-1}

\bibliography{references}

\onecolumngrid
\appendix

\section{MQB at $\a=1/2$}
\label{appa}

Here derive the distribution \eqref{eq:distributionnk} in the main text for the macroscopic variable
\begin{align}
    X_{\a=1/2}^{[N]} = \frac{\sum_{i=1}^{N} (a_i - \mu)}{\sqrt{N}\t}
\end{align}
on a general superposition state $\ket{\Psi_N} = \sum_k c_k \ket{N,k}$. The characteristic function is
\begin{align}
    \chi(t) & = \sum_{a_1, ..., a_{N}} \text{tr} \bigg[ \rho \Big( E_{a_1}^{(1)} \ox ... \ox E_{a_{N}}^{(N)} \Big) \bigg]e^{i t X^{[N]}} \nonumber \\
    & =  \, \text{tr} \bigg[ \rho \Big( \sum_{a_1} E_{a_1}^{(1)} e^{i t (a_1-\m) / (\sqrt{N} \t)} \Big) \ox ... \ox  \Big( \sum_{a_{N}} E_{a_{N}}^{(N)} e^{i t (a_{N}-\m) / (\sqrt{N} \t)} \Big) \bigg] \nonumber \\
    & =   \sum_{k,l}c_k^* c_l   \, \bra{N, k} \bigg[ \Big( \sum_a E_a e^{i t (a-\m) / (\sqrt{N} \t)} \Big)^{\ox N} \bigg] \ket{N, l} \nonumber \\
    & =   \sum_{k,l}c_k^* c_l   \, \bra{N, k} \mathcal{A}^{\ox N} \ket{N, l} \, ,
\end{align}
where $\mathcal{A} = \sum_a E_a e^{i t (a-\m)/(\sqrt{N} \t)}$. For clarity, let us first compute the matrix element in the sum above for the case $k \geq l$:
\begin{align}
    \bra{N, k} & \, \mathcal{A}^{\ox N} \,  \ket{N, l} = \nonumber \\
    & = \frac{1}{\sqrt{ \binom{N}{k} \binom{N}{l}}} \Big( \bbra{\underbrace{1...1}_{k} 0...0} + \text{perm.} \Big) \mathcal{A}^{\ox N} \Big( \bket{\underbrace{1...1}_{l} 0...0} +  \text{perm.} \Big) \nonumber \\
    & = \sqrt{\frac{\binom{N}{k}}{ \binom{N}{l}}}  \, \bbra{\underbrace{1...1}_{k} 0...0}  \, \mathcal{A}^{\ox N} \Big( \bket{\underbrace{1...1}_{l} 0...0} +  \text{perm.} \Big) \nonumber \\
    & =  \sqrt{\frac{\binom{N}{k}}{ \binom{N}{l}}}  \, \sum_{q=0}^l  \, \binom{k}{q} \binom{N-k}{l-q}  \, \bbra{\underbrace{1...1}_{k} 0...0}  \, \mathcal{A}^{\ox N} \bket{\underbrace{\underbrace{1...1}_{q}0...0}_{k} \underbrace{\underbrace{1..1}_{l-q} 0...0}_{N-k}}  \nonumber \\
    & = \sqrt{\frac{\binom{N}{k}}{ \binom{N}{l}}}  \,  \sum_{q=0}^l \,  \binom{k}{q} \binom{N-k}{l-q}  \,  \mel{1}{ \mathcal{A} }{1}^{q} \, \mel{1}{\mathcal{A} }{0}^{k-q} \, \mel{0}{ \mathcal{A} }{1}^{l-q} \,  \mel{0}{\mathcal{A} }{0}^{N-k-l+q} \nonumber \\
    & = \mel{0}{ \mathcal{A}}{0}^{N-k-l} \,  \sum_{q=0}^l  \frac{\sqrt{k! \, l! \,}}{ q! (k-q)! (l-q)!}\, N^{\frac{k+l-2q}{2}} \big[ 1 + O(1/N) \big] \,  \mel{1}{ \mathcal{A}}{1}^{q} \, \mel{1}{ \mathcal{A} }{0}^{k-q}  \, \mel{0}{ \mathcal{A} }{1}^{l-q} \,  \mel{0}{\mathcal{A} }{0}^{q} \, .
\end{align}
In the second equality we have used permutational invariance; in the third equality we have gathered all the terms that contribute equally, multiplied by their combinatorial multiplicity; in the last line we have used Stirling's formula. Now, let us expand 
\begin{align}
    \mathcal{A} & =  \sum_a E_a \Big[ 1 + i t \frac{a-\m}{\sqrt{N}\t} - \frac{t^2}{2} \frac{(a-\m)^2}{N \t^2} + O(N^{-3/2}) \Big] \nonumber \\
    & = 1 + \frac{it}{\sqrt{N}\t} \Big( \sum_a a E_a-\m \Big) - \frac{t^2}{2 N \t^2} \Big( \sum_a a^2 E_a - 2 \m \sum_a a E_a + \m^2 \Big) + O(N^{-3/2}) \nonumber \\
    & = 1 + \frac{it}{\sqrt{N}\t} \big( A -\m \big) - \frac{t^2}{2 N \t^2} \big( A^{(2)} - 2 \m A + \m^2 \big) + O(N^{-3/2}) \, , 
\end{align}
where for convenience we have introduced the matrices $A=\sum_a a E_a$ and $A^{(2)} = \sum_a a^2 E_a$. Then, we can see that $\mel{1}{\mathcal{A}}{1}^q = 1 + O(1/\sqrt{N})$ and $\mel{0}{\mathcal{A}}{0}^q = 1 + O(1/\sqrt{N})$, while $\mel{1}{\mathcal{A}}{0}^{k-q} = \Big( \frac{i t A_{10}}{\sqrt{N} \t} \Big)^{k-q} \big[ 1 + O(1/\sqrt{N}) \big]$ and $\mel{0}{\mathcal{A}}{1}^{l-q} = \Big( \frac{i t A_{01}}{\sqrt{N} \t} \Big)^{l-q} \big[ 1 + O(1/\sqrt{N}) \big]$. These first order contributions to the off-diagonal matrix elements cancel the overall factor of $N^{\frac{k+l-2q}{2}}$, while higher order corrections are suppressed. On the other hand, the factor in the front is
\begin{align}
    \mel{0}{\mathcal{A}}{0}^{N-k-l} & = \exp \Big\{ (N-k-l) \log \mel{0}{\mathcal{A}}{0} \Big\} \nonumber \\
    & = \exp \bigg\{ (N-k-l) \log \bigg[ 1 + \frac{it (A_{00}-\m)}{\sqrt{N} \t}- \frac{t^2 (A^{(2)}_{00} - 2 \m A_{00} + \m^2)}{2N\t^2} + O(N^{-3/2}) \bigg] \bigg\} \nonumber \\
    & = \exp \bigg\{ (N-k-l) \bigg[ \frac{i t (A_{00}-\m)}{\sqrt{N} \t} - \frac{t^2 (A^{(2)}_{00}-2 \m A_{00}+\m^2)}{2N\t^2} + \frac{t^2 (A_{00}-\m)^2}{2N\t^2} + O(N^{-3/2}) \bigg] \bigg\} \nonumber \\
    & = \exp \bigg\{ \frac{i t \sqrt{N}}{\t}(A_{00}-\m) - \frac{t^2 \s^2}{2 \t^2} +O(1/\sqrt{N}) \bigg\} \, , 
\end{align}
where we have further defined $\s^2 = A^{(2)}_{00} - (A_{00})^2$. Then, the matrix element reads
\begin{align}
    \mel{N,k}{\mathcal{A}^{\ox N}}{N,l} = e^{i t \sqrt{N} (A_{00}-\m)/\t - t^2 \s^2/(2 \t^2) + O(1/\sqrt{N})} \sum_{q=0}^l \frac{\sqrt{k! \, l!}}{q! (k-q)! (l-q)!} \left( \frac{i t A_{10}}{\t} \right)^{k-q}   \left( \frac{i t A_{01}}{\t} \right)^{l-q} \big[ 1 +O(1/\sqrt{N}) \big] \, .
\end{align}
For the case $l>k$ all we have to do is exchange $k$ and $l$, so that the sum only runs until the smallest of the two, and also exchange $A_{10}$ and $A_{01}$. Then, choosing $\m = A_{00}$ and $\t = |A_{01}|$, the characteristic function in the limit reads
\begin{align}
    \chi(t) = \sum_{k,l} e^{-ik\j} c_k^* c_l e^{il\j} \,  e^{- t^2 \s^2/(2 \t^2)} \sum_{q=0}^{\min(k,l)} \frac{\sqrt{k! \, l!}}{q! (k-q)! (l-q)!}   (- i t)^{k+l-2q} \, ,
    \label{eq:chi0}
\end{align}
where $\j = \arg(-A_{01})$. In order to obtain the probability distribution of $X^{[N]}$ we take the Fourier transform $\mathcal{F}[\chi(t)](x)$ of the previous expression. Using that and $\mathcal{F} \big[ (it)^n f(t) \big] (x) = \frac{d^n}{dx^n} \mathcal{F} \big[  f(t) \big] (x)$, we get
\begin{align}
    P ( x ) & =   \sum_{k,l} e^{-i k \j} c_k^* c_l e^{i l \j}  \,  \sum_{q=0}^{\min(k,l)}   \frac{\sqrt{k! \, l! \,}}{q! (k-q)! (l-q)!}   \, \frac{(-1)^{k+l-2q}}{ \sqrt{2 \pi} \s/\t}  \, \frac{d^{k+l-2q}}{dx^{k+l-2q}} \, e^{ - \t^2 x^2 /(2 \s^2) } \nonumber \\
    & = \sum_{k,l} e^{-i k \j} c_k^* c_l e^{i l \j}  \,    \sum_{q=0}^{\min(k,l)}  \frac{\sqrt{k! \, l! \,}}{q! (k-q)! (l-q)!}   \, \frac{1}{ \sqrt{2 \pi} \s/\t} \, \left(   \frac{\t}{\s} \right)^{k+l-2q} \,  e^{-\left( \frac{\t x}{\sqrt{2} \s} \right)^2} \, \he_{k+l-2q} \left( \frac{\t x}{\s} \right) \nonumber \\
    & = \sum_{k,l}e^{-i k \j} c_k^* c_l e^{i l \j}  \, \frac{e^{-\left( \frac{\t x}{\sqrt{2} \s} \right)^2}}{ \sqrt{2 \pi}  \s/\t}   \sum_{q=0}^{\min(k,l)} \frac{\sqrt{k! \, l! \,}}{q!} {k+l-2q \choose l-q} \left( \frac{\t}{\s} \right)^{k+l-2q} \, \frac{  \he_{k+l-2q} \left( \frac{\t x}{\s} \right)}{ (k+l-2q)!} \, ,
    \label{eq:distributionbeforelemma}
\end{align}
where in the second equality we have used Rodrigues' formula for Hermite polynomials \cite{abramowitz} and in the last equality we have introduced the binomial coefficient for convenience. Now, the above sum over Hermite polynomials can be written as a product of two Hermite polynomials of order $k$ and $l$ by virtue of the following lemma:

\begin{lemma*}
For any constants $\a$, $\b$ and $\c$ satisfying $\a^2=\b^2+\c^2$, and for any non-negative integers $m$ and $n$ with $n \leq m$, the following identity of Hermite polynomials holds:
\begin{align}
    \frac{e^{-\left( \frac{x}{\sqrt{2} \a} \right)^2}}{\sqrt{2 \pi} \a} \sum_{s=0}^n &  \frac{1}{s!} {m+n-2s \choose n-s} \left( \frac{\c}{\a} \right)^{m+n-2s} \, \frac{\he_{m+n-2s} \left( \frac{x}{\a}\right)}{(m+n-2s)!}  = \int dx'  \frac{e^{-\left( \frac{x-x'}{\sqrt{2} \b} \right)^2}}{\sqrt{2 \pi} \b}  \frac{e^{-\left( \frac{x'}{\sqrt{2} \c} \right)^2}}{\sqrt{2 \pi} \c} \, \frac{\he_m \left( \frac{x'}{ \c} \right)}{ m!}  \frac{\he_n \left(\frac{x'}{\c} \right)}{n!} \, .
\end{align}
\end{lemma*}
\begin{proof}
Let us first define the generalized Hermite polynomials as done by Nielsen \cite{nielsen}:
\begin{align}
    \ghe_n(x,a) = \frac{(-1)^n (2a)^n}{n!} \, e^{\frac{x^2}{4a}} \, \frac{d^n}{dx^n} e^{- \frac{x^2}{4a}} \, .
\end{align}
These are related to the ones used in the main text by
\begin{align}
    \ghe_n \left( x, 1/2 \right) = \frac{1}{ n!} \, \he_n (x) \, ,
\end{align}
and they satisfy the following properties
\begin{align}
    \ghe_n ( b \, x, b^2 \, a) & = b^n \, \ghe_n(x,a) \, , \quad  \forall b \in \mathbb{R} \, , \\
    \int_{-\infty}^{+\infty} dx \, \ghe_n(x,a) e^{- \xi x^2 + \eta x} & = \sqrt{\frac{\pi}{\xi} \, } \,  e^{\frac{\eta^2}{4 \xi}} \, \ghe_n \left( \frac{\eta}{2 \xi} , a - \frac{1}{4 \xi} \right)  \, , \quad  \forall \eta, \xi \in \mathbb{R} \, \text{ with } \,  \xi > 0 \, .
\end{align}
As proved by Nielsen, the generalized Hermite polynomials satisfy the identity
\begin{align}
  \ghe_m \left( \frac{x}{\c}, a \right) \ghe_n \left(\frac{x}{\c}, a \right) = \sum_{s=0}^n \frac{(2a)^s}{s!} {m+n-2s \choose n-s}  \, \ghe_{m+n-2s} \left( \frac{x}{\c}, a \right)   \, , 
\end{align}
which, setting $a=\frac{1}{2}$, corresponds to the limit $\b =0$ of the wanted identity. Taking the Weierstrass transform and using the above properties yields
\begin{align}
    \int_{-\infty}^{+\infty} dx'  \frac{e^{-\left( \frac{x-x'}{\sqrt{2} \b} \right)^2}}{\sqrt{2 \pi} \b}  \frac{e^{-\left( \frac{x'}{\sqrt{2} \c} \right)^2}}{\sqrt{2 \pi} \c} & \, \ghe_m \left( \frac{x'}{\c}, a \right) \ghe_n \left(\frac{x'}{\c}, a \right) = \nonumber \\
    & = \int_{-\infty}^{+\infty} dx'  \frac{e^{-\left( \frac{x-x'}{\sqrt{2} \b} \right)^2}}{\sqrt{2 \pi} \b}  \frac{e^{-\left( \frac{x'}{\sqrt{2} \c} \right)^2}}{\sqrt{2 \pi} \c} \,  \sum_{s=0}^n \frac{(2a)^s}{s!} {m+n-2s \choose n-s}  \, \ghe_{m+n-2s} \left( \frac{x'}{\c}, a \right) \nonumber \\
    & = \frac{e^{-\left( \frac{x}{\sqrt{2} \b} \right)^2}}{2 \pi \b} \sum_{s=0}^n \frac{(2a)^s}{s!} {m+n-2s \choose n-s} \int_{-\infty}^{+\infty} du \, e^{- \frac{\a^2}{2 \b^2} u^2 + \frac{\c x}{\b^2} u} \, \ghe_{m+n-2s} (u,a) \nonumber \\
    & =  \frac{e^{-\left( \frac{x}{\sqrt{2} \b} \right)^2}}{2 \pi \b} \sum_{s=0}^n \frac{(2a)^s}{s!} {m+n-2s \choose n-s} \sqrt{\frac{2 \pi \b^2}{\a^2}} \, e^{\left( \frac{\c x}{ \sqrt{2} \a \b} \right)^2} \, \ghe_{m+n-2s} \left( \frac{\c x}{\a^2} ,  a - \frac{\b^2}{2 \a^2} \right) \nonumber \\
    & = \frac{e^{-\left( \frac{x}{\sqrt{2} \a} \right)^2}}{\sqrt{2 \pi} \a} \sum_{s=0}^n \frac{(2a)^s}{s!} {m+n-2s \choose n-s}   \, \left( \frac{\c}{\a} \right)^{m+n-2s} \,  \ghe_{m+n-2s} \left( \frac{x}{\a} ,  \frac{1}{2} + (2a-1) \frac{\a^2}{2 \c^2} \right) \, .
\end{align}
Setting $a = 1/2$ completes the proof.
\end{proof}
Then, the probability distribution \eqref{eq:distributionbeforelemma} may be written as
\begin{align}
    P(x) & = \sum_{k,l} e^{-i k \j} c_k^* c_l e^{i l \j}  \,  \sqrt{k! \, l! \,}   \int_{-\infty}^{+\infty} dx' \frac{e^{- \left( \frac{x-x'}{\sqrt{2} s} \right)^2}}{\sqrt{2 \pi} s} \frac{e^{- \left( \frac{x'}{\sqrt{2}} \right)^2}}{\sqrt{2 \pi}} \, \frac{\he_k \left( x' \right)}{  k!} \frac{\he_l \left( x' \right)}{ l!} \nonumber  \\
    & = \sum_{k,l} e^{-i k \j} c_k^* c_l e^{i l \j}  \,   \int_{-\infty}^{+\infty} dx' \frac{e^{- \left( \frac{x-x'}{\sqrt{2} s}  \right)^2}}{\sqrt{2 \pi} s}  \left[ \frac{ e^{-x^2/4} \, \he_k (x')}{ (2 \pi)^{1/4} \sqrt{k!}} \right]^* \,  \left[ \frac{ e^{-x^2/4} \, \he_l (x')}{ (2 \pi)^{1/4} \sqrt{k!}} \right] \nonumber \\
    & = \sum_{k,l} e^{-i k \j} c_k^* c_l e^{i l \j}  \,   \int_{-\infty}^{+\infty} dx' \frac{e^{- \left( \frac{x-x'}{\sqrt{2} s}  \right)^2}}{\sqrt{2 \pi} s} \braket{k}{x'} \braket{x'}{l} \, ,
\end{align}
where $s^2=\s^2/\t^2-1$ and we have introduced the wave-functions
\begin{align}
    \braket{x}{k} = \frac{1}{(2 \pi)^{1/4}} \frac{1}{\sqrt{k!}} \, e^{-x^2/4} \, \he_k (x) 
\end{align}
of a one dimensional harmonic oscillator with constant $m \omega/\hbar=1/2$. This proofs the formula \eqref{eq:distributionnk} from the main text.

\section{Violation of ML at $\a=1/2$}
\label{appb}
Here we show a simple Bell violation with a correlated state of the form
\begin{align}
    \ket \psi = \sum_k c_k \ket k _A \ox \ket k _B \, .
\end{align}
Consider the CHSH inequality
\begin{align}
-2 \leq \avg{\mathcal{B}} \leq 2 \, ,
\end{align}
in terms of the Bell-CHSH parameter $\mathcal{B} = A \ox (B+B')+A'\ox (B-B')$, where $A, A', B$ and $B'$ are $\pm 1$-valued observables. Here we set these to be the sign of phase-space observables, i.e.
\begin{align}
    A= U^\dagger_{\j_A} \sgn(x_A) U_{\j_A} \, , \quad A'= U^\dagger_{\j_A'} \sgn(x_A) U_{\j_A'} \, ,
\end{align}
in terms of the unitaries $U_{\j} = \sum_k e^{i k \j} \ket{k} \bra{k}$, and similarly for $B$ and $B'$. The quantum correlations are of the form
\begin{align}
    \avg{AB}_Q & = \sum_{k,l} e^{-ik(\j_A+\j_B)} c_k^* c_l e^{i l(\j_A + \j_B)} \, _A \mel{k}{\sgn(x_A)}{l}_A \, \, _B\mel{k}{\sgn(x_B)}{l}_B \nonumber \\
    & = \sum_{k,l} e^{-ik(\j_A+\j_B)} c_k^* c_l e^{i l(\j_A + \j_B)} \, (i_{kl})^2 \, ,
\end{align}
where 
\begin{align}
    i_{kl} = \int_{-\infty}^{+\infty} dx \, \sgn (x) \, \braket{k}{x} \, \braket{x}{l} \, .
\end{align}
Let
\begin{align}
\ket{\psi} =\frac{2}{\sqrt{10}} \,  \ket{0}_A  \ket{0}_B + \frac{1}{\sqrt{2}}\, \ket{1}_A  \ket{1}_B + \frac{1}{\sqrt{10}}  \ket{2}_A \ket{2}_B  \, ,
\end{align}
Then the only non-zero integrals are $i_{01} = i_{10} = \sqrt{2/\pi}$ and $i_{12} = i_{21} = 1/\sqrt{\pi}$, and the expectation value of the Bell-CHSH parameter can be written as
\begin{align}
    \avg{\mathcal{B}}_Q = \frac{\sqrt{5}}{\pi} \big(  c_{11} + c_{12} + c_{21} - c_{22} \big)
\end{align}
where $c_{ij} = \cos \big( \j_A^{(i)} + \j_B^{(j)} \big)$. Now, the standard arrangement for the maximal CHSH violation (with the angles separated by $45^{\circ}$), the combinations of cosines attains its maximum value $2 \sqrt{2}$, and so the Bell inequality is violated:
\begin{align}
\avg{\mathcal{B}}_Q = 2 \frac{\sqrt{10}}{\pi} >2 \, .
\end{align}

\section{Robustness of the MQB at $\a=1/2$}
\label{appc}
\subsubsection*{Robustness against particle loss}
Here we show robustness of the previous result against particle loss. Let us assume a simple model in which we associate to every particle a random variable $o_i\in \{0,1\}$, where $o_i=0$ represents the event when the $i$-th particle is lost and $o_i=1$ when it is not lost, thus reaching the detectors. Let us assume that these random variables are independent, with probability of reaching the detectors $P(o_i=1)=p \in [0,1]$ and probability of being lost $P(o_i=0) =1-p$. Then the intensity measured is $\sum_{i=1}^N o_i a_i$. Let us further assume that one can count the number of particles received, so that the resulting macroscopic variable (rescaled by a factor of $p$ as discussed in the main text)  is
\begin{align}
    X_p^{[N]} = \frac{\sum_{i=1}^N o_i (a_i - \m)}{p \, \t \sqrt{N}} \, .
\end{align}
The corresponding characteristic function is
\begin{align}
    \chi_p(t) & = \sum_{o_1, ..., o_N} \sum_{a_1, ..., a_N} e^{i t X^{[N]}} P(o_1, ..., o_N) \, \tr \bigg\{ \rho \Big( E_{a_1} \ox ... \ox E_{a_N} \Big)\bigg\} \nonumber \\ 
    & = \tr \bigg\{ \rho \bigg[ \sum_{a_1} E_{a_1} \Big( 1 - p +p e^{it(a_1-\m)/(p \t \sqrt{N})} \Big) \bigg] \ox ... \ox \bigg[\sum_{a_N} E_{a_n} \Big( 1 - p +p e^{it(a_N-\m)/(p \t \sqrt{N})} \Big) \bigg] \bigg\} \nonumber \\ 
    & = \tr \Big\{ \rho \, \mathcal{A}_p^{\ox N} \Big\}
\end{align}
where $\mathcal{A}_p = \sum_a E_a \Big( 1 - p +p e^{it(a-\m)/(p \t \sqrt{N})} \Big)$. The calculation is thus very similar to the one in the previous section. Choosing as before $\m=A_{00}$ and $\t = |A_{01}|$, the characteristic function in the limit reads
\begin{align}
    \chi_p(t) = \sum_{k,l} e^{-ik\j} c_k^* c_l e^{il\j} \, \exp \Big\{ - \frac{t^2 \s^2}{2 p^3 \, \t^2} \Big\} \, \sum_{q=0}^{\min(k,l)} \frac{\sqrt{k! \, l!}}{q! (k-q)! (l-q)!} (-it)^{k+l-2q}  \, ,
\end{align}
where again $\j = \arg(-A_{01})$. We see that all that has changed with respect to \eqref{eq:chi0} is that we have $\s^2/(p^3 \t^2)$ instead of $\s^2/\t^2$. In the end, this simply translates into an enlarged width of the Gaussian POVM \eqref{eq:gaussianpovm}
\begin{align}
    s_p^2 = \frac{\s^2}{p^3 \, \t^2} - 1 \, .
\end{align}

\subsubsection*{Other noise channels}
Let us further consider some independent noise channels acting on individual particles. For instance, consider the depolarizing and dephasing channels
\begin{align}
    \G_\text{pol}(\rho) & = (1-\l) \rho + \l \frac{\mathds{1}}{2} \, , \\
    \G_\text{pha}(\rho) & = (1-\l) \rho + \l Z \rho Z \, ,
\end{align}
where $\l \in [0,1]$ and $Z$ is the phase flip. Applying such channels to the density matrix is equivalent to applying the adjoint channel to the one-particle POVM
\begin{align}
    \G_\text{pol}^\dagger(E_a) & = (1-\l) E_a + \l \, \frac{\mathds{1}}{2} \,  \tr E_a \, , \\
    \G_\text{pha}^\dagger(E_a) & = (1-\l) E_a + \l Z E_a Z \, .
\end{align}
These are still POVMs. Then, all that needs to be modified in the macroscopic limit is the parameters $s$ and $\j$ that specify the macroscopic POVM in \eqref{eq:linearmap}. In particular we find
\begin{align}
    s_\text{pol}^2 & = \frac{(1-\l) A_{00}^{(2)} + \frac{\l}{2} \tr A^{(2)} - \Big( (1-\l) A_{00} + \frac{\l}{2} \tr A \Big)^2}{(1-\l)^2 \t^2 }- 1 \, , \quad & \j_\text{pol} = \j + \frac{1- \sgn (1- \l)}{2} \pi \, , \, \, ~ \nonumber \\
    s_\text{pha}^2 & = \frac{\s^2}{ (1-2 \l)^2 \t^2} -1 \, , \quad & \j_\text{pha}  = \j + \frac{1 - \sgn (1- 2 \l)}{2} \pi \, ,
\end{align}
where $A= \sum_a a E_a$, $A^{(2)} = \sum_a a^2 E_a$, $\s^2 = A_{00}^{(2)} - (A_{00})^2$, $\t = |A_{01}|$ and $\j = \arg(-A_{01})$.

\section{MQB at $\a=1$}
\label{appd}
Here we derive the distribution for the macroscopic variable
\begin{align}
    X_{\a=1}^{[2N]} = \frac{1}{\t \, 2 N} \sum_{i=1}^{2N} \big(a_i - \m \big)
\end{align}
on a general superposition state $\ket{\Psi_{2N}} = \sum_k c_k \ket{2N, N+k}$, and we also show that the corresponding bipartite distributions can be explained by a local model. The characteristic function of $X^{[2N]}$ is
\begin{align}
    \chi(t) &= \sum_{a_1, ..., a_{2N}} e^{i t X^{[2N]}} \tr \Big[ \rho \big( E_{a_1} \ox ... \ox E_{a_{2N}} \big)  \Big] \nonumber \\ 
    & = e^{-it \m / \t} \tr \bigg[ \rho \Big( \sum_{a_1} E_{a_1} e^{i t a_1/(2 N \t)} \Big) \ox ... \ox \Big( \sum_{a_{2N}} E_{a_{2N}} e^{i t a1/(2 N \t)} \Big) \bigg] \nonumber \\
    & = e^{-i t \m / \t} \sum_{k,l} c_k^* c_l \, \mel{2N,N+k}{\mathcal{A}^{\ox 2N} }{2N, N+l} \, ,
\end{align}
where $\mathcal{A} = \sum_a E_a e^{ita/(2N \t)}$. Now, using the following expression for Dicke states
\begin{align}
    \ket{2N, N+k} = \frac{1}{\sqrt{{2N \choose N+k}}} \frac{1}{2 \pi} \int_{-\pi}^\pi d\f \big( \ket 0 + e^{i\f} \ket 1 \big)^{\ox 2N} \, e^{-i(N+k)\f} \, ,
\end{align}
the matrix element in the sum above is
\begin{align}
    & \mel{2N,N+k}{ \mathcal{A}^{\ox 2N} }{2N, N+l} = \nonumber \\
    & = \frac{1}{(2 \pi)^2} \frac{1}{\sqrt{{2N \choose N+k} {2N \choose N+l}}} \int_{-\pi}^\pi d\f_1 d\f_2 \,  e^{i (N+k)\f_1} \big( \bra 0 + e^{-i\f_1} \bra 1  \big)^{\ox 2N} \mathcal{A}^{\ox 2N} \big( \ket 0 + e^{i\f_2} \ket 1  \big)^{\ox 2N}   e^{-i(N+l)\f_2} \nonumber \\
    & = \frac{1}{(2 \pi)^2} \frac{1}{\sqrt{{2N \choose N+k} {2N \choose N+l}}}\int_{-\pi}^\pi d\f_1 d\f_2 \,  e^{i (N+k)\f_1 - i (N+l)\f_2} \Big[  \big( \bra 0 + e^{-i\f_1} \bra 1  \big) \mathcal{A} \big( \ket 0 + e^{i\f_2} \ket 1  \big)  \Big]^{2N}\nonumber \\ 
    & = \frac{1}{(2 \pi)^2} \int_{-\pi}^\pi d\f_1 d\f_2 \,  e^{i (k\f_1 - l\f_2)} \frac{e^{i N (\f_1 -\f_2) }}{\sqrt{{2N \choose N+k} {2N \choose N+l}}} \big[ \mathcal{A}_{00} + e^{i\f_2} \mathcal{A}_{01} + e^{-i\f_1} \mathcal{A}_{10} + e^{i(\f_2-\f_1)} \mathcal{A}_{11} \big]^{2N} \nonumber \\
    & = \frac{1}{(2 \pi)^2} \int_{-\pi}^\pi d\f_1 d\f_2 \, e^{i (k\f_1 - l\f_2)} \frac{e^{i N (\f_1 -\f_2) } \big[ 1 + e^{i(\f_2-\f_1)} \big]^{2N}}{\sqrt{{2N \choose N+k} {2N \choose N+l}}} \bigg[ \frac{\mathcal{A}_{00} + e^{i\f_2} \mathcal{A}_{01} + e^{-i\f_1} \mathcal{A}_{10} + e^{i(\f_2-\f_1)} \mathcal{A}_{11}}{1 + e^{i(\f_2-\f_1)}} \bigg]^{2N}
\end{align}
Now, introducing the expansion $\mathcal{A}_{ij} = \d_{ij} + \frac{i t}{2N \tau} A_{ij} + O(N^{-2})$ in terms of the matrix elements of $A= \sum_a a E_a$, we have
\begin{align}
    & \mel{2N,N+k}{ \mathcal{A}^{\ox 2N} }{2N, N+l} = \nonumber \\
    & = \frac{1}{(2 \pi)^2} \int_{-\pi}^\pi d \f_1 d\f_2 \, e^{i (k \f_1 - l \f_2)} \frac{ \Big( e^{i(\f_1-\f_2)} \big[ 1 + e^{i (\f_2-\f_1)} \big]^2 \Big)^N}{\sqrt{{2N \choose N+k} {2N \choose N+l}}} \times \nonumber \\ 
    & \qquad \qquad \qquad  \qquad \qquad  \qquad \qquad  \qquad \qquad  \times \bigg[ \frac{1 + \frac{it}{2N\t} A_{00}+ e^{i \f_2} \frac{it}{2N\t} A_{01} + e^{-i \f_1} \frac{it}{2N\t} A_{10} + e^{i(\f_2-\f_1)} (1 + \frac{it}{2N\t} A_{11} )}{1 + e^{i(\f_2-\f_1)}} \bigg]^{2N} \nonumber \\
    & = \frac{1}{(2 \pi)^2} \int_{-\pi}^\pi d \f_1 d\f_2 \, e^{i (k \f_1 - l \f_2)} \frac{ \Big( 2 \big[ 1 + \cos(\f_1-\f_2) \big] \Big)^N}{\sqrt{{2N \choose N+k} {2N \choose N+l}}} \bigg[ 1 + \frac{i t}{2N\t} \frac{ A_{00} + e^{i \f_2}  A_{01}+ e^{-i \f_1} A_{10} + e^{i(\f_2-\f_1)} A_{11}}{1 + e^{i(\f_2-\f_1)}} \bigg]^{2N} \nonumber \\
    & = \frac{1}{(2 \pi)^2} \int_{-\pi}^\pi d \f_1 d\f_2 \, e^{i (k \f_1 - l \f_2)} \frac{4^N \cos^{2N} \left( \frac{\f_1-\f_2}{2} \right)}{\sqrt{{2N \choose N+k} {2N \choose N+l}}} \bigg[ 1 + \frac{i t}{2N\t} \frac{ A_{00} + e^{i \f_2}  A_{01}+ e^{-i \f_1} A_{10} + e^{i(\f_2-\f_1)} A_{11}}{1 + e^{i(\f_2-\f_1)}} \bigg]^{2N}
\end{align}
Using the limit representation of Dirac delta function
\begin{align}
    \d(x) =\frac{1}{ 2 \pi} \lim_{N \to \infty} \frac{ 4^N \, \cos^{2N} (x/2)}{{2N \choose N}} \, , \quad x \in [-\pi, \pi]\, ,
\end{align}
the limit yields
\begin{align}
    \lim_{N \to \infty} & \mel{2N,N+k}{ \mathcal{A}^{\ox 2N} }{2N, N+l} =\nonumber \\
    & = \frac{1}{(2 \pi)^2} \int_{-\pi}^\pi d \f_1 d\f_2 \, e^{i (k \f_1 - l \f_2)} 2 \pi \d ( \f_1 - \f_2) \exp \Big\{ \frac{i t}{\t} \frac{ A_{00} + e^{i \f_2}  A_{01}+ e^{-i \f_1} A_{10} + e^{i(\f_2-\f_1)} A_{11}}{1 + e^{i(\f_2-\f_1)}} \Big\} \nonumber \\
    & = \frac{1}{2 \pi} \int_{-\pi}^\pi d \f \, e^{i (k-l) \f}  \exp \Big\{ \frac{i t}{\t} \frac{ A_{00} + e^{i \f}  A_{01}+ e^{-i \f} A_{10} +  A_{11}}{1 + 1} \Big\} \nonumber \\
    & = \exp \Big\{ \frac{it}{\t}  \frac{\tr A}{2} \Big\} \frac{1}{2 \pi} \int_{-\pi}^\pi d \f \, e^{i (k-l) \f}  \exp \Big\{ \frac{i t}{2\t} \big( e^{i \f}  A_{01}+ e^{-i \f} A_{10} \big) \Big\} \, .
\end{align}
Then, going back to the characteristic function and choosing $\mu = \frac{1}{2} \tr A$ and $A_{01} = \t e^{i \varphi}$, we have
\begin{align}
    \chi(t) &= \sum_{k,l} c_k^* c_l \frac{1}{2 \pi} \int_{-\pi}^\pi d\f \, e^{i(k-l) \f} \exp \Big\{ \frac{i t }{2} \big(e^{i (\f + \j)} + e^{-i(\f + \j)} \big) \Big\} \nonumber \\ 
    & =\sum_{k,l} c_k^* c_l  \frac{1}{2 \pi} \int_{-\pi}^\pi d\f \, e^{ i (k-l) \f +  i t   \cos (\f + \j) } \nonumber \\ 
    & =\sum_{k,l} c_k^* c_l  \frac{1}{2 \pi} \int_{-\pi}^\pi d\f \, e^{ i (k-l) (\f-\j) +  it  \cos \f } \, .
\end{align}
And the probability distribution is
\begin{align}
    P(x) & = \sum_{k,l} c_k^* c_l \frac{1}{(2 \pi)^2} \int_{-\pi}^\pi d\f \, e^{i(k-l)(\f-\j)} \int dt \, e^{-ixt} \, e^{ it \cos \f } \nonumber \\
    & =\sum_{k,l}  e^{-i k \j} \, c_k^* c_l \,  e^{i l \j}  \,  \frac{1}{2 \pi } \int_{-\pi}^\pi d\f \, e^{i(k-l) \f} ~ \d \left( \cos \f - x \right) \, . \label{eq:probintphi}
\end{align}
Now, if we introduce $x= \cos \q$, we have
\begin{align}
    P(\q) & = \frac{dx}{d\q}  \, \sum_{k,l} e^{-i k \j} \, c_k^* c_l \,  e^{i l \j}  \,   \frac{1}{2 \pi } \int_{-\pi}^\pi d\f \, e^{i(k-l) \f} ~ \d \left( \cos \f - \cos \q \right) \nonumber \\ 
    & = \sum_{k,l}  e^{-i k \j} \, c_k^* c_l \,  e^{i l \j}  \,   \, \frac{e^{i(k-l) \q} + e^{-i(k-l) \q}}{2 \pi } \, .
\end{align}
This proves equation \eqref{eq:distribution2nn} from the main text. Now consider a bipartite scenario, where Alice and Bob measure the macroscopic variables $X_{\a=1}^{[2N]}$ and $Y_{\a=1}^{[2N]}$ respectively on a state of the form
\begin{align}
    \ket{\Psi_{4N}} = \sum_{k,l} c_{kl} \ket{2N, N+k}_A \ox \ket{2N, N+l}_B \, .
\end{align}
Then, using \eqref{eq:probintphi}, it is easy to see that the bipartite distribution can be written straightforwardly as
\begin{align}
    P(x, y) = \sum_{klmn} e^{-i k \j_A - i l \j_B} c_{kl}^* c_{mn} e^{im\j_A + in\j_B} \frac{1}{(2 \pi)^2} \int_{-\pi}^\pi d\f_1 d\f_2 \, e^{i(k-m) \f_1 + i (l-n)\f_2} \, \d(\cos \f_1 - x) \, \d(\cos \f_2 - y) \, .
\end{align}
This is a local model:
\begin{align}
    P(x, y) =  \int_{-\pi}^\pi d\f_1 d\f_2 \, \bigg| \frac{1}{2 \pi}  \sum_{k,l} e^{i k (\j_A - \f_1)  + i l (\j_B - \f_2)} \, c_{kl} \bigg|^2  \, \d(\cos \f_1 - x) \, \d(\cos \f_2 - y) \, ,
\end{align}
where the term squared is the distribution $\m(\l)$ of the local hidden variable $\l = (\f_1, \f_2)$.

\end{document}